\documentclass[12pt]{article}

\usepackage[titletoc,title]{appendix}
\usepackage[latin1]{inputenc}
\usepackage{amsmath}
\usepackage{amssymb}
\usepackage{graphicx}
\usepackage[margin=1.25in]{geometry}
\usepackage[bottom]{footmisc}
\usepackage{indentfirst}
\usepackage{endnotes}
\usepackage{mathabx}
\usepackage{enumerate}
\usepackage{rotating}
\usepackage{multirow}
\usepackage{booktabs,calc}
\usepackage{threeparttable}
\usepackage{threeparttable}
\usepackage{float}
\usepackage{url}
\usepackage{epsfig}
\usepackage{caption,subcaption}
\usepackage{float}
\usepackage{url}
\usepackage{amsmath,amsfonts}
\usepackage{amsthm}
\usepackage{mathrsfs}

\usepackage[onehalfspacing]{setspace}
\usepackage{apptools}

\newcommand{\suchthat}{\;\ifnum\currentgrouptype=16 \middle\fi|\;}

\newtheorem{prop}{Proposition}
\newtheorem{lemma}{Lemma}

\newtheorem{defn}{Definition}

\theoremstyle{remark}

\usepackage{chngcntr}
\usepackage{graphicx}
\usepackage{amsmath}
\usepackage{amssymb}
\usepackage{parskip}
\usepackage{sgame}
\usepackage{color}
\usepackage{tikz}
\usetikzlibrary{trees,calc}

\DeclareMathOperator*{\argmax}{\arg\!\max}

\newcommand{\eps}{\varepsilon}

\newcommand{\R}{\mathbb R}

\usepackage{natbib}

\begin{document}

\title{\large{Nontransitive Preferences and Stochastic Rationalizability: A Behavioral Equivalence}\thanks{We thank P.J. Healy for helpful comments and advice. John Rehbeck is supported in part by the National Science Foundation grant NSF-2049749. Opinions, findings, conclusions, or recommendations offered here are those of the authors and do not necessarily reflect the views of the National Science Foundation. Mogens Fosgerau has received funding from the European Research Council (ERC) under the European Union's Horizon 2020 research and innovation programme (grant agreement No. 740369).  Any errors are our own.}}

\author{
    \small{Mogens Fosgerau} \\
    \small{Department of Economics} \\
    \small{University of Copenhagen} \\
    \small{mogens.fosgerau@econ.ku.dk}\\
    \and
    \small{John Rehbeck}  \\
    \small{Department of Economics} \\
    \small{The Ohio State University} \\
    \small{rehbeck.7@osu.edu}\\
}
\date{\small{ \today }}

\maketitle
\begin{abstract}
Nontransitive choices have long been an area of curiosity within economics. However, determining whether nontransitive choices represent an individual's preference is a difficult task since choice data is inherently stochastic. This paper shows that  behavior from nontransitive preferences under a monotonicity assumption is equivalent to a transitive stochastic choice model. In particular, nontransitive preferences are regularly interpreted as a strength of preference, so we assume alternatives are chosen proportionally to the nontransitive preference. One implication of this result is that one cannot distinguish ``complementarity in attention" and ``complementarity in demand." 
\end{abstract}

\newpage
\setlength\parindent{24pt}
\section{Introduction}
Economics has long been interested in understanding how the context of a decision problem affects choice. However, once one allows  context to affect choice, it is easy to find revealed preferences that are nontransitive.\footnote{One example of nontransitive revealed preference is the Condorcet paradox \citep{condorcet1785essai} where a voter chooses between pairs of candidates in a way that cycles through candidates. This thought experiment was studied empirically by \cite{may1954intransitivity} where alternatives were imbued with attributes that differed in various dimensions and found revealed preference cycles. A general theory of deterministic nontransitive preferences for binary choice was characterized in \cite{shafer1974nontransitive}. Greater interest to study the psychological motivation behind nontransitive choices followed the regret theory of \cite{loomes1982regret}. A more modern model that allows nontransitive preferences is the salience theory of \cite{bordalo2012salience}.} While nontransitive preferences have garnered some interest in economics, they are unsettling for those accustomed to utility theory where the value of an object is well defined. Moreover, testing models of nontransitive preferences is difficult since models of nontransitive preferences are often deterministic while the data used to test the models is inherently stochastic. Thus, many papers following \cite{tversky1969intransitivity} have looked for violations of stochastic transitivity  to find evidence in support of nontransitive preferences.\footnote{An example of stochastic transitivity is weak stochastic transitivity, which states that for a triple of alternatives, $x,y,z \in X$, where $p(x,y)$ is the probability $x$ is chosen from the binary menu $\{x,y\}$, if $p(x,y) \ge \frac{1}{2}$ and $p(y,z)\ge \frac{1}{2}$, then $p(x,z) \ge \frac{1}{2}$.}

 This paper shows when choice is stochastic that it is not possible to distinguish behavior from a large class of nontransitive preference models from a model of transitive stochastic choice. The only restrictions we make on nontransitive preferences are that they are continuous in the value of an alternative and satisfy a monotonicity assumption. The monotonicity assumption generalizes the requirement that even though preferences are nontransitive, if an alternative is made better off holding all else fixed, then that alternative should be chosen more often. This result is important because there are now many papers that try to distinguish nontransitive preferences only using choice data such as \cite{tversky1969intransitivity}, \cite{loomes1991observing}, \cite{humphrey2001non}, and \cite{birnbaum2010testing}. The equivalence of nontransitive and transitive stochastic choice shows that one needs an additional source of data to distinguish nontransitive preferences. For example, one might use decision times as in \cite{alos2022identifying} or brain imaging data as in \cite{kalenscher2010neural}.

 The behavioral equivalence we show also sheds some light on different interpretations of stochastic choice behavior. For example, there are several experimental studies showing that improving attribute values of one alternative can cause the choice probability of a different alternative to go up \citep{huber1982adding, simonson1989choice}.\footnote{Examples of this are the ``attraction effect" and ``compromise effect."} This is an economic puzzle since models of additive random utility (e.g. logit \citep{mcfadden1973conditional}) impose the property that if an alternative is made better, then all other choice probabilities must decrease. Broadly, the experimental evidence  shows that alternatives can be ``complements" in stochastic choice. Here, we say there are complements since increasing the desirability of one alternative increases the probability a different alternative is chosen.
 
 There are two natural ways to allow complementarity in stochastic choice. One natural way to allow complementarity is via ``complementarity in attention" where the value of an alternative depends on how it is compared to the available alternatives. This approach is exemplified in the salience functions of \cite{bordalo2012salience} and \cite{bordalo2013salience}. The second way to allow complementarity is through ``complementarity in demand" by allowing a non-linear preference term that generalizes the standard additive error as discussed in \cite{allen2019revealed} and \cite{allen2020hicksian}. This complementarity occurs in perturbed utility models as exemplified by the inverse product differentiation logit model of \cite{fosgerau2022inverse}. Our result shows that any complementarity in attention can be represented as complementarity of demand and vice-versa. Therefore, these types of complementarity are behaviorally indistinguishable. 
 
 We now discuss how we model nontransitive preferences. While there are many ways to model nontransitive preferences that are fairly structured, we take a relatively reduced form approach. We do this for a few reasons. First, many models of nontransitive preference are designed only for binary choice settings. Our approach is flexible enough to allow for many alternatives. Second, we focus on how the value of alternatives interact rather than modeling the attribute values directly. We do this since it is difficult to identify relevant attributes ex-ante. Third, since we do not explicitly model attributes, we avoid complications of modelling the interactions of attributes. For example, \cite{fishburn1984multiattribute, bouyssou1986some}, and \cite{vind1991independent} assume that attributes interact in an additive separable way when making a decision. In contrast, we assume there is just some function of the alternative values that governs choice.    
 
 We now discuss more formally what we mean by a transitive and nontransitive model of choice in the deterministic setting.\footnote{Here, we use nontransitive choice to also refer to choice that might be transitive. This is similar to how nonexpected utility theory contains expected utility as a special case.} We assume that each alternative $a$ has some utility relevant value $v_a \in \mathbb{R}$.\footnote{One can think of the value of an alternative being elicited via a multiple price list or a willingness to pay.} We allow the value of each alternative to be in the real numbers since we assume there are attributes such as price or quality that affect this value but are not modeled. For a standard deterministic preference maximizer, an individual would choose an alternative from the menu that has the highest value $v_a$. For the nontransitive choice model, we assume for each alternative $a$ there is a function $T_a( (v_b)_{b\in A} )$ that takes the value of all alternatives in menu $A$ and assigns a value to alternative $a$. Thus, any perceived similarity of difference between attribute values of alternatives, reference dependence, or other psychological effects are not explicitly modeled. The values taken by the $T_a$ terms are sometimes interpreted as a strength of preference such as in \cite{bell1982regret} and \cite{dyer1982relative}. The difference between the deterministic transitive and nontransitive models can be summarized where 
 \begin{center}
 \begin{tabular}{l l}
        Transitive choice: & $ \argmax_{a\in A} \{ v_a \}$  \\
       Nontransitive choice:  & $ \argmax_{a \in A} \{T_a( (v_b)_{b\in A} )\}$ .
\end{tabular}
\end{center}

So far, we have only discussed how we have modeled preferences. However, it is worthwhile discussing some issues when trying to take models of deterministic preference to inherently stochastic data. One way to account for stochastic data is by introducing  random variables $\eps=(\eps_a)_{a \in A}$ and $\eps'=(\eps'_a)_{a \in A}$ and studying choices generated according to 
\[ \argmax_{a\in A} \{ v_a + \eps_a \} \quad \text{and} \quad \argmax_{a \in A} \{ T_a( (v_b)_{b\in A} ) + \eps'_a \}. \]
From the above maximization process, there are several measurement problems. First, $T_a$ and $v_a$ are in different units. This means if one makes an assumption on the error term  $\eps'$, then it is possible that some other error term $\eps$ generates the same choice probabilities. Second, the $T_a$ functions can be transformed without changing the deterministic ordering, but these transformations can change the choice probabilities. Throughout the paper, we make a high level psychological assumption on how $T_a$ generate choice probabilities to avoid introducing error terms. In detail, since the $T_a$ are often interpreted as a strength of preference for alternative $a$ so we assume that alternatives are chosen proportionally to $T_a$. This is justified by psychological research that examines perceptual errors such as \cite{thurstone1927law} and \cite{luce1959individual}.\footnote{Perceptual errors are those that occur when there is some objectively measurable amount that can be learned but is not known when the individual answers a question. For example, asking an individual ``Which of these two stones weighs more?" when they do not have access to a scale to measure the true weights is an example where perceptual errors occur.} 

Now that we have discussed how we map nontransitive preferences to choice probabilities, we discuss our notion of behavioral equivalence. Within the paper, we look at whether a transitive perturbed utility model (PUM) \citep{mcfadden2012theory,allen2019revealed} can rationalize stochastic choices generated by nontransitive preferences. In detail, we look for a transitive perturbed utility model where an individual chooses probabilities $p=(p_a)_{a \in A}$ that are non-negative and sum to one, there is a cost function $C(p)$, and choice probabilities are generated by maximizing
\[ \sum_{a \in A} v_a p_a - C(p). \]
We call this a transitive PUM since the values of alternatives enter linearly and do not interact with one another. Perturbed utility models assume that individuals have a  preference to randomize as suggested in \cite{machina1985stochastic}. It is important to allow a preference for randomization since there is growing descriptive evidence that supports this viewpoint.\footnote{For example, \cite{agranov2017stochastic} consecutively repeat several choice problems and find individuals randomize, \cite{feldman2022revealing} find individuals choose mixtures of lotteries that are positively correlated with repeated decisions, and \cite{agranov2020stable} find individuals choose randomly in similar ways when playing games and making individual decisions.} 
This class of models also generalizes additive random utility models so it can account for stochastic behavior arising from either a preference to randomize or additive random shocks.\footnote{To see that any additive random utility model can be written as a PUM, see \cite{hofbauer2002global}, \cite{allen2019identification}, or \cite{sorensen2022mcfadden}.}

This paper contributes to the large literature in economics looking at when models can be distinguished from observable behavior. One classic example is the long literature following \cite{samuelson1938note} that began a further study of whether preferences and utility theory could be empirically distinguished via choice behavior. Some more modern examples looking at behavioral equivalences include \cite{matejka2014rational} and \cite{masatlioglu2016behavioral}. \cite{matejka2014rational} show that logit can be thought of as resulting from information acquisition by an agent with limited attention.\footnote{\cite{Fosgerau2016r} generalizes this result and shows any additive random utility model can be thought of as resulting from information acquisition by an agent with limited attention.} \cite{masatlioglu2016behavioral} show that the models of stochastic reference dependence, quadratic preferences, and probability weighting are all behaviorally equivalent under certain restrictions and how one can behaviorally distinguish between these models.

This paper also contributes to the literature on nontransitive preferences. Before discussing the literature in more detail, there are many papers written by Peter Fishburn on nontransitive preferences (e.g., \cite{fishburn1982nontransitive,fishburn1984multiattribute, fishburn1988context,fishburn1989non,fishburn1990unique} ) including an insightful survey \citep{fishburn1991nontransitive}. Throughout several of Fishburn's papers, one thing mentioned in passing is that nontransitive preferences might be able to be thought of as transitive in some richer class of preferences over lotteries. We formalize this conjecture under a monotonicity assumption.

For a broader review of nontransitive preferences, some of the first work to examine nontransitive preferences was in \cite{katzner1971demand}, \cite{sonnenschein1971demand}, and  \cite{shafer1974nontransitive}. These results are general and respectively examine how to evaluate demand, whether equilibrium exist, and representation theorems for nontransitive preferences.\footnote{Other papers that may be of interest to the reader are revealed preference characterizations of nontransitive preferences in \cite{john2001concave}, \cite{quah2006weak}, and \cite{keiding2013revealed}.} Greater interest in nontransitive preference followed after the formalization of regret theory in \cite{bell1982regret} and \cite{loomes1982regret}. In addition, nontransitive choices can occur in theories of multiattribute decision making such as the salience theory in \cite{bordalo2012salience} and \cite{bordalo2013salience}. Both regret theory and salience theory have generated large literatures looking at psychological motivations of behavior and whether these models can yield better predictions. Our paper contributes to this literature by showing how models of nontransitive preference when measured with noisy data can be rationalized by a transitive stochastic choice model.

The rest of the paper proceeds as follows. Section~\ref{sec:setup} describes in more detail the mathematical setup, the assumptions made on nontransitive preferences, and defines the transitive perturbed utility model. Section~\ref{sec:results} provides the main results of the paper. Section~\ref{sec:conclusion} gives our final remarks.

\section{Setup}\label{sec:setup}
We begin with the deterministic model of nontransitive preferences. Suppose there are $X$ alternatives and we consider a nonempty subset $A \subseteq X$. Our analysis focuses on one subset of alternatives since the same analysis can be repeated for each subset. It also allows us to avoid further indexing of the subset in the arguments. We assume there is a continuous function  $T:\mathbb{R}^{A} \rightarrow \mathbb{R}_{+}^{A} \setminus \{0\}$ where the function $T$ allows the values of alternatives to interact with one another depending on the menu $A$.\footnote{Ruling out $T(v) = 0$ assumes there is always some alternative (not necessarily fixed for all $v \in \mathbb{R}^A$) that is perceived as a good.} Here, the value $T_{a}(v)$ is a non-negative value of alternative $a$ in the presence of values in the vector $v = (v_a)_{a \in A} \in \mathbb{R}^A$. We allow the values of $v_a \in \mathbb{R}$ since we assume there are attributes that affect the value but are not modeled so that alternative values can be arbitrarily large or small.\footnote{For example, alternatives are often modeled to have values $v_a = q_a - c_a$ where $q_a$ is a measure of quality and $c_a$ is the cost of the alternative. Here $q_a$ captures all attribute interactions and the cost of an alternative $c_a$ is linear. Thus, values in $\mathbb{R}$ mean cost and quality can be arbitrarily large. The results do not depend on this assumption.} One way to elicit these values would be to use a multiple price list or a willingness to pay measure without the context of other alternatives. For example, see \cite{somerville2022range} and \cite{im2023multiple} for eliciting values for goods with attributes. 

We now define how to go from the deterministic setting to a stochastic setting. We assume that $p(v) = (p_a(v))_{a \in A}$ are choice probabilities. Thus, for each $a \in A$ we have $p_a(v) \ge 0$ and $\sum_{a \in A} p_a(v)=1$. We also represent the set of probability distributions over the alternatives by  $\Delta = \{ p \in \mathbb{R}_+^A \mid \sum_{a \in A} p_a = 1 \}.$

When going to a stochastic setting from deterministic nontransitive preferences, we make an intuitive assumption that agrees with the interpretation of nontransitive preferences. The terms $T_a(v)$ are regularly interpreted as a ``strength of preference" such as in \cite{bell1982regret} and \cite{dyer1982relative}. Thus, we assume the $T_a(v)$ terms map to a stochastic setting proportional to the strength of preference. This is consistent with the Weber-Fechner law where physical stimuli are perceived proportionally which motivated the work of \cite{thurstone1927law} and \cite{luce1959individual}. Therefore, we define the \emph{normalized nontransitive preference} that generates stochastic choice. 

\begin{defn}\label{def:nontransitivepreference}
For a nontransitive preference $T:\mathbb{R}^{A} \rightarrow \mathbb{R}_{+}^{A}$, we assume that choice probabilities are proportional to the strength of preference. Therefore, we define the \emph{normalized nontransitive preference} $\tilde{T}: \mathbb{R}^A \rightarrow \Delta$ where for each $a \in A$
\[ p^{\tilde{T}}_a(v) = \tilde{T}_a(v) = \frac{T_a(v)}{\sum_{b \in A} T_b(v)} .\]
\end{defn}

The normalized nontransitive preference preserves the ordering of the nontransitive preference for every $v \in \mathbb{R}^{A}$. To see this, note that 
\[ T_a(v) \ge T_b(v) \; \Leftrightarrow \; \frac{T_a(v)}{\sum_{c \in A} T_c(v)} \ge \frac{T_b(v)}{\sum_{c \in A} T_c(v)} \; \Leftrightarrow \; \tilde{T}_a(v) \ge \tilde{T}_b(v) .   \]
Therefore, from the perspective of deterministic preference maximization this assumption is without loss of generality. This transformation is also helpful since we now have a unitless measure of strength of preference. The normalized nontransitive preference also immediately gives a mapping into choice probabilities where $p^{\tilde{T}}(v) = \tilde{T}(v)$. Lastly $\tilde{T}$ is continuous since it is a continuous transform of the deterministic nontransitive preferences. 

At this point, the theory of nontransitive preferences does not have any empirical content. To see this, note that for any given choice probability function $p(v)$, we can set construct a nontransitive preference $T(v) = p(v)$, which leads to  $p$ as the stochastic choice according to Definition~\ref{def:nontransitivepreference}. This means that any behavior can be represented via a nontransitive preference unless further conditions are imposed. Continuity of $p$ is not essential for this argument.

To make the model testable, we place a restriction on the normalized nontransitive preferences $\tilde{T}$, namely that they satisfy a general monotonicity condition. The monotonicity condition we place on $\tilde{T}$ is the multivariate generalization of monotonicity. This condition is similar to the strong axiom of revealed preference (see for example \cite{Richter1966}), but rather than being ordinal it is cardinal since we are working with measured values.  One convenient feature of the monotonicity condition is that it implies the normalized nontransitive preference for an alternative is increasing in the value of that alternative holding all else fixed.

\begin{defn}
The normalized nontransitive preference satisfies cyclic monotonicity when for any sequence $\{v^i\}_{i=1}^I$ with $v^i \in \mathbb{R}^{A}$ that 
\[ \sum_{i=1}^I \sum_{a \in A} \tilde{T}_a(v^i)(v_a^{i} - v_a^{i+1}) \ge 0 \]
where $v^{I+1}=v^1$.
\end{defn}

This feature assumes that as an alternative has a higher value holding all else fixed, then it will gain choice probability. For example, choose an alternative $a'$ and consider values $v=(v_{a'},v_{-a'})$ and $v'=(v'_{a'},v_{-a'})$ where $v_{-a'}$ are values of alternatives other than $a'$ and $v_{a'} \neq v'_{a'}$. Applying cyclic monotonicity, we get that 
\[ (\tilde{T}_{a'}(v) - \tilde{T}_{a'}(v')) (v_{a'} - v'_{a'}) \ge 0  \]
so that increasing the value of the alternative makes it relatively more appealing. We note that it is important to work with the normalized values rather than the direct nontransitive preference $T$ since an increase in the level of some $T_a$ does not necessarily imply an increase in the relative strength of preference to other alternatives. Cyclic monotonicity can be tested by examining whether a normalized nontransitive preference satisfies this monotonicity with elicited values via a linear programming problem similar to the one in \cite{allen2019revealed}. 

We now discuss our notion of a stochastic rationalization by a transitive perturbed utility model (PUM) \citep{mcfadden2012theory, allen2019revealed}. We focus on the transitive perturbed utility model since it encompasses both additive random utility models and a preference for randomization. We require the transitive perturbed utility model to have a unique choice probability that maximizes an individual's preferences. 

\begin{defn}\label{def:PUM}
We say a stochastic choice function $p(v)$ is rationalized by a \emph{transitive perturbed utility model} when there exists a convex function $C:\Delta \rightarrow \mathbb{R}$ such that for every $v \in \mathbb{R}^A$ the choice function is uniquely given by 
\[ p(v) = \argmax_{p \in \Delta} \left\{ \sum_{a \in A} v_a p_a - C(p) \right\} .\]
\end{defn}

We briefly mention a few points about the rationalization. We call this a transitive perturbed utility model since there is no direct interaction between the $v$ terms in this model. This means that any kind of complementarity or substitutability arises from the convex function $C$. Therefore, we have shut off ``complementarity in attention" but allow ``complementarity in demand." Our main result will show that these two notions of complementarity are equivalent. Therefore, ``complementarity in attention" and ``complementarity in demand" capture have the same observable restrictions on behavior from choice data. If one takes a stronger stance on complementarity and incorporates other observables such as decision times as in \cite{alos2022identifying} or neuronal activity as in \cite{kalenscher2010neural}, then it may be possible to distinguish between these two types of complementarity.  

\section{Main Results}\label{sec:results}
We now present our main result. We show that any normalized nontransitive preference that is cyclically monotone is equivalent to a transitive PUM. Here, the equivalence means that any normalized nontransitive preference that is cyclically monotone can be represented by a transitive perturbed utility model and vice-versa. We record the result below and give a proof.

\begin{prop}
    A normalized nontransitive preference satisfies cyclic monotonicity if and only if it can be rationalized by a transitive PUM.
\end{prop}
\begin{proof}
Suppose that the normalized nontransitive preference $\tilde{T}(v)$ is cyclically monotone. By \citet[][Thm. 24.8]{rockafellar1970convex}, there is a closed proper convex function $f$ such that $\tilde{T}(v) \in \partial f(v)$ where $\partial f$ is the sub-differential of $f$.\footnote{The subdifferential of $f$ at $v$ is the set of all subgradient vectors of $f$ at $v$. Here a vector $w^*$ is a subgradient of $f$ at $v$ when $f(z) \ge f(v) + \sum_{a\in A }w_a^*(z_a-v_a)$ for all $z$.}  By Lemma \ref{lem:maximalCM} in the Appendix, the correspondence $v\rightarrow \{ \tilde T (v)\}$ is maximal and hence $\partial f (v)=\{ \tilde T (v) \},\forall v \in \R^{A} $.

From the convex function $f$, we define the convex conjugate of $f$, 
\[ f^*(w^*) = \sup_{ w \in \mathbb{R}^{A} } \left\{\sum_{a \in A} w_a w_a^* - f(w) \right\}\]
which is a convex function by construction. We show that $f^*$ defines a perturbation function that gives a transitive PUM  rationalization.

Consider the unconstrained transitive perturbed utility model of 
\begin{equation}\label{eq:PUM}
   \sup_{w^* \in \R^A} \left\{ \sum_{a \in A} v_a w_a^* - f^{*}(w^*) \right\} . 
\end{equation}
This is maximized at $w^* = \tilde{T}(v)$ by an application of \citet[][Thm. 23.5]{rockafellar1970convex} since $f$ is a closed proper convex function. To see this, note that $w^*$ attains the maximum if and only if $v \in \partial f^*(w^*)$, and which occurs if and only if $w^* \in \partial f(v)$. Finally, this means that $w^* = \tilde{T}(v)$ since $\partial f = \{\tilde{T}(v)\}$. Thus, we can restrict the domain of the supremum in \eqref{eq:PUM} to $\Delta$ without changing the optimizer that is obtained. Therefore, the transitive perturbed utility model of Definition \ref{def:PUM} with $C=f^{*}$ rationalizes the normalized nontransitive preferences $\tilde{T}$.

Conversely, suppose that choices are generated by a transitive perturbed utility model so that 
\[ p^*(v) = \argmax_{p \in \Delta} \{ \sum_{a \in A} p_a v_a - C(p) \}.\]
We have that $p^*(v)$ is continuous via the theorem of the maximum ( see for example \citet[][Thm. 9.17]{sundaram1996first}), since $p^*(v)$ is assumed to be a unique maximizer.

It follows that $\tilde{T}(v)=p^*(v)$ is a normalized nontransitive preference that is cyclically monotone. The argument essentially follows from \citet[][Thm. 24.8]{rockafellar1970convex}, but we reproduce the argument here. For any cycle $v^i,i=1,\ldots, I$ with $v^I = v^0$, we have for each $i$ that   
\[ \sum_{a \in A} p_{a}^*(v^{i}) v_a^i - C(p^*(v^i)) \ge \sum_{a \in A} p_{a}^*(v^{i+1}) v_a^i - C(p^*(v^{i+1})) ,\]
since $p^*(v^i)$ maximizes the transitive perturbed utility at $v^i$.
Rearranging shows that
\[ \sum_{a \in A} (p_{a}^*(v^{i}) - p_a^*(v^{i+1})) v_a^i \ge C(p^*(v^i)) - C(p^*(v^{i+1})).\]
Summing this over all $i$, we get that 
\[ \sum_{i=1}^I \sum_{a \in A} (p_{a}^*(v^{i}) - p_a^*(v^{i+1})) v_a^i \ge 0.\]
Finally, this last inequality can be rearranged to 
\[ \sum_{i=1}^I \sum_{a \in A} p_{a}^*(v^i) (v_a^i - v_a^{i-1})\ge 0.\]
Since the cycle is arbitrary, the choices are cyclically monotone. 
\end{proof}

\section{Conclusion}\label{sec:conclusion}
This paper shows a behavioral equivalence between a model of normalized nontransitive choice and transitive stochastic choice under a monotonicity assumption. Thus, we formalize the conjecture of Peter Fishburn that nontransitive choice might be transitive in a richer class of preferences. In detail, we assume that the nontransitive preference is a strength of preference as suggested by \cite{bell1982regret} and \cite{dyer1982relative} so that how likely an alternative is chosen is proportional to the strength of preference. This is motivated by the assumption of proportional discernment for physical stimuli (e.g. difference in weights) that is often known as the Weber-Fechner law. It is worthwhile to note that early stochastic choice modeling of \cite{thurstone1927law} and \cite{luce1959individual} were also motivated by the observation that physical stimulai are discerned proportionally. Therefore, this paper helps link the literature on nontransitive preferences and transitive stochastic choice.

One implication of the behavioral equivalence of normalized nontransitive preferences and transitive stochastic choice is that one cannot distinguish ``complementarity in attention" and ``complementarity in demand." In other words, one cannot distinguish between complementarity that occurs due to context effects and complementarity that occurs from a preference for randomization \citep{machina1985stochastic}. This result follows since we do not place assumptions on how attributes interact and work directly with values. If one places stronger assumptions on how the attributes enter, then it might be possible to differentiate these two effects. Alternatively, one way to identify information on nontransitive preferences is to use additional non-choice data such as decision times \citep{alos2022identifying} or brain imaging data \citep{kalenscher2010neural}. However, in these cases one would need to make additional assumptions on how nonchoice data links to attention or a preference for randomization.  

\section{Appendix}
\begin{lemma}\label{lem:maximalCM}
    Let $T:\R^A \rightarrow \R^A$ be a continuous function such that the single-valued correspondence $v\rightarrow \{ T(v)\}$ is cyclically monotone. 
    Then $v\rightarrow \{ T(v)\}$  is maximal cyclically monotone.
\end{lemma}
\begin{proof}
    We will apply  \citet[][Thm. 5.2]{Crouzeix2010}. With their notation, $S=D=\R^A$ and we can let $V=\R ^A$. For any $a\in D, N_D(a)={0}$. Moreover, translating their definition to the present context, $\Gamma_c(a,S)=\{T(a)\}$, since $T$ is continuous. This implies by their Theorem 5.2 that the correspondence $v\rightarrow \{T(v)\}$ is maximal cyclically monotone.
\end{proof}

\bibliographystyle{ecta}
\bibliography{ref,mf}

\end{document}